\documentclass[12pt]{article}
\usepackage{fullpage,amsthm}
\usepackage{amsmath,amssymb}
\usepackage{color}
\usepackage{graphicx}
\usepackage{booktabs}
\usepackage{enumerate}
\usepackage{caption}
\usepackage{times}

\newcommand{\spread}{{\rm spread}}

\newcommand{\ceil}[1]{\lceil #1 \rceil}
\newcommand{\floor}[1]{{\lfloor #1 \rfloor}}

\newtheorem{lemma}{Lemma}
\newtheorem{proposition}{Proposition}
\newtheorem{theorem}{Theorem}

\newtheorem{remark}{Remark}
\begin{document}

\title{Diffuse Reflection Radius in a Simple Polygon}
\author{
  Eli Fox-Epstein\thanks{Department of Computer Science, Brown University, Providence, RI. \texttt{ef@cs.brown.edu}} \and
  Csaba D. T\'oth\thanks{Department of Mathematics, California State University, Northridge, Los Angeles, CA. \texttt{cdtoth@acm.org}} \and
  Andrew Winslow\thanks{D\'{e}partement d'Informatique, Universit\'{e} Libre de Bruxelles, Brussels, Belgium. \texttt{awinslow@ulb.ac.be}}
}
\date{}
\maketitle

\begin{abstract}
It is shown that every simple polygon in general position with $n$ walls can be illuminated from a single point light source $s$ after at most $\lfloor (n-2)/4\rfloor$ \emph{diffuse reflections}, and this bound is the best possible. A point $s$ with this property can be computed in $O(n\log n)$ time.
It is also shown that the minimum number of diffuse reflections needed to illuminate a given simple polygon from a single point can be approximated up to an additive constant in polynomial time.
\end{abstract}

\section{Introduction}
When light diffusely reflects off of a surface, it scatters in all directions.
This is in contrast to specular reflection,
  where the angle of incidence equals the angle of reflection.
We are interested in the minimum number of diffuse reflections needed to illuminate all points in the interior of a simple polygon $P$ with $n$ vertices from a single light source $s$ in the interior of $P$.
A \emph{diffuse reflection path} is a polygonal path $\gamma$ contained in $P$ such that every interior vertex of $\gamma$ lies in the relative interior of some edge of $P$, and the relative interior of every edge of $\gamma$ lies in the interior of $P$ (see Fig.~\ref{fig:diffuse-ex1} for an example). Our main result is the following.

\begin{theorem}\label{thm:radius}
For every simple polygon $P$ with $n\geq 3$ vertices in general position (i.e., no three collinear vertices), there is a point $s\in {\rm int}(P)$ such that for all $t \in {\rm int}(P)$,
  there is a diffuse reflection path from $s$ to $t$ with at most $\floor{(n-2)/4}$ internal vertices.
This lower bound is the best possible.
A point $s\in {\rm int}(P)$ with this property can be computed in $O(n \log n)$ time.
\end{theorem}

\begin{figure}[htb]
\centering
\includegraphics[width=\textwidth]{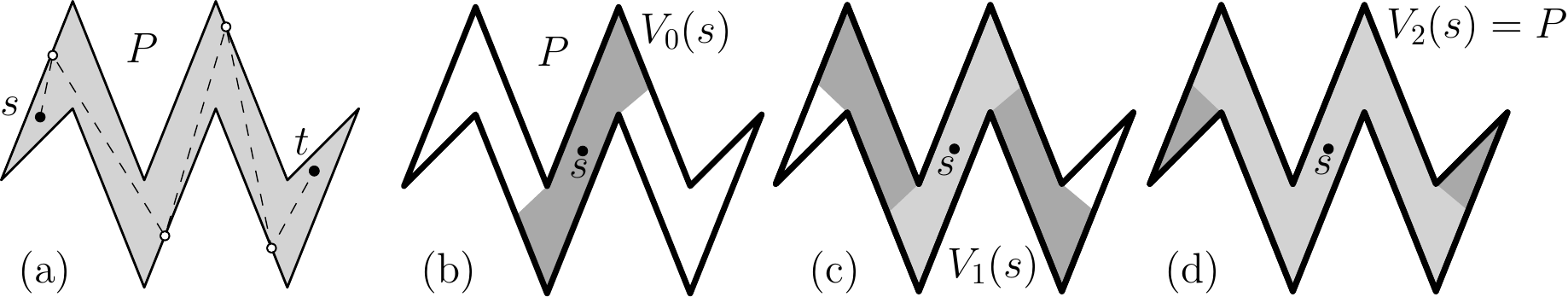}
\caption{%
(a) A diffuse reflection path between $s$ to $t$ in a simple polygon $P$. (b)--(d) The regions of a polygon illuminated by a light source $s$ after 0, 1, and 2 diffuse reflections. The diffuse reflection radius of a zig-zag polygon with $n$ vertices is $\floor{(n-2)/4}$.}
\label{fig:diffuse-ex1}
\end{figure}

This result is, in fact, a tight bound on the worst-case diffuse reflection radius (defined below) for simple polygons.
Denote by $V_k(s)\subseteq P$ the part of the polygon illuminated by a light source $s$ after at most $k$ diffuse reflections. Formally, $V_k(s)$ is the set of points $t\in P$ such that there is a diffuse reflection path from $s$ to $t$ with at most $k$ interior vertices. Hence, $V_0(s)$ is the visibility polygon of point $s$ within the polygon $P$ if $s\in {\rm int}(P)$. The \emph{diffuse reflection depth} of a point $s\in {\rm int}(P)$ is the smallest integer $r\geq 0$ such that ${\rm int}(P)\subseteq V_r(s)$. The \emph{diffuse reflection radius} $R(P)$ of a simple polygon $P$ is the minimum diffuse reflection depth over all points $s\in {\rm int}(P)$, and \emph{diffuse reflection center} of $P$ is the set of points $s\in {\rm int}(P)$ that attain this minimum. With this terminology, Theorem~\ref{thm:radius} implies that $R(P)\leq \floor{(n-2)/4}$ for every simple polygon $P$ with $n\geq 3$ vertices in general position. A family of zig-zag polygons (e.g.\ the polygon in Fig.~\ref{fig:diffuse-ex1}) shows that this bound is the best possible for all $n\geq 3$. We note here that the \emph{diffuse reflection diameter} $D(P)$ of $P$ is the \emph{maximum} diffuse reflection depth over all $s\in {\rm int}(P)$.

No polynomial-time algorithm is known for computing $R(P)$ for a given polygon $P$ with $n$ vertices. We show, however, that $R(P)$ can be approximated up to a constant additive error in polynomial time.
\begin{theorem}\label{thm:apx-compute-radius}
 Given a simple polygon $P$ with $n$ vertices in general position,
 one can compute in time polynomial in $n$:
\begin{enumerate}
\item an integer $k\in \mathbb{N}_0$ such that $k-1\leq R(P)\leq k+1$, and
\item a point $s\in {\rm int}(P)$ such that ${\rm int}(P)\subseteq V_{k+1}(s)$.
\end{enumerate}
\end{theorem}

\paragraph{\bf Motivation and Related Work.}
Diffuse reflection paths are special cases of \emph{link paths}, which have been studied extensively due to applications in motion planning, robotics, and curve compression~\cite{G07,MSD00}.
The \emph{link distance} between two points, $s$ and $t$, in a simple polygon $P$ is the minimum number of edges in a polygonal path between $s$ and $t$ that lies entirely in $P$.
In a polygon $P$ with $n$ vertices, the link distance between two points can be computed in $O(n)$ time~\cite{S86}.
The \emph{link depth} of a point $s$ is the smallest integer $d\geq 0$ such that all other points in $P$ are within link distance $d$ of $s$.
The \emph{link radius} is the minimum link depth over all points in $P$, and the \emph{link center} is the set of points with minimum link depth.
It is known that the link center is a convex region and can be computed in $O(n\log n)$ time~\cite{DLS92}.
The \emph{link diameter} of $P$, the maximum link depth over all points in $P$, can also be computed in $O(n\log n)$ time~\cite{S90}.

The \emph{geodesic center} of a simple polygon is a point inside the polygon which minimizes the maximum shortest-path distance  (also known as geodesic distance) to any point in the polygon. Asano and Toussaint~\cite{AT85} proved that the geodesic center is unique. Pollack et al.~\cite{PSR89} show how to compute the geodesic center of a simple polygon with $n$ vertices in $O(n \log n)$ time; this was recentpy improved to $O(n)$ time by Ahn et al.~\cite{ABB+15}. Hershberger and Suri~\cite{HS97} give an $O(n)$ time algorithm for computing the \emph{geodesic diameter}. Schuirer~\cite{Sch94} gives $O(n)$ time algorithms for the geodesic center and diameter under the $L_1$ metric in rectilinear polygons. Bae et al.~\cite{BKOW14} show that the $L_1$-geodesic diameter and center can be computed in $O(n)$ time in every simple polygon with $n$ vertices.

Diffuse reflection paths have received increasing attention since the mid-1990s when Tokarsky~\cite{Tok95} answered a question of Klee~\cite{Kle69,Kle79}, proving that a light source may not cover the interior of the simple polygon using \emph{specular} reflection (where the angle of incidence equals the angle of reflection in the reflection path).
He constructed a simple polygon $P$ and two points $s,t\in {\rm int}(P)$ such that there is no specular reflection path from $s$ to $t$. It is not difficult to see that all points $t\in {\rm int}(P)$ can be reached from any $s\in {\rm int}(P)$ on a diffuse reflection path. However, the maximum number of reflections, in terms of the number of vertices, have been determined only recently. Barequet et al.~\cite{Us} proved, confirming a conjecture by Aanjaneya et al.~\cite{ABP08}, that $D(P)\leq \lfloor n/2\rfloor-1$ for all simple polygons with $n$ vertices, and this bound is the best possible.

The link distance, geodesic distance and the $L_1$-geodesic distance are all metrics; but the minimum number of edges on a diffuse reflection path between two points is \emph{not} a metric. Specifically, the triangle inequality need not hold (note that for $a,b,c\in {\rm int}(P)$, the concatenation of an two diffuse reflection paths, $a$-to-$b$ and $b$-to-$c$, need not be a diffuse reflection path since it may have an interior vertex at $b\in {\rm int}(P)$). This explains, in part, the difficulty of handling diffuse reflections.
Brahma et al.~\cite{BPS04} constructed examples where $V_2(s)$ (the set of points reachable from $s$ after at most two diffuse reflections) is not simply connected, and where $V_3(s)$ has $\Omega(n)$ holes. In general, the maximum complexity of $V_k(s)$ is known to be $\Omega(n^2)$ and $O(n^9)$~\cite{ADI+06}. In contrast to link paths, the best known algorithm for computing a minimum diffuse reflection path (one with the minimum number of reflections) between two points in a simple polygon with $n$ vertices takes $O(n^9)$ time~\cite{ADI+06,G07}.
Ghosh et al.~\cite{GGM+12} give a 3-approximation for this problem that runs in $O(n^2)$ time. Bishnu et al.~\cite{BGG+14} define a \emph{constrained} version of diffuse reflection paths that can be computed in $O(n^3)$ time.

Khan et al.~\cite{KPA+13} study two weaker models of diffuse reflections, in which some edges of a diffuse reflection path may overlap with the boundary of the polygon $P$. They establish upper and lower bounds for
the diffuse reflection radius under these weaker models for simple polygons that can be decomposed into convex quadrilaterals. No previous bound has been known for the diffuse reflection radius under the standard model that we use in this paper.

\paragraph{\bf Proof Technique.}
The regions $V_k(s)$ are notoriously difficult to handle. Instead of $V_k(s)$, we rely on the simply connected regions $R_k(s)\subseteq V_k(s)$ defined by Barequet et al.~\cite{Us} and show that ${\rm int}(P)\subseteq R_{\floor{(n-2)/4}}(s)$ for some point $s\in {\rm int}(P)$. In Section~\ref{sec:prelim}, we establish a simple sufficient condition (Lemma~\ref{lem:condition}) for ${\rm int}(P)\subseteq R_{\floor{(n-2)/4}}(s)$ in terms of the visibility polygon $V_0(s)$ that can be verified in $O(n)$ time. Except for two extremal cases that are resolved directly (Section~\ref{ssec:double}), we prove that there \emph{exists} a point satisfying these conditions in Section~\ref{sec:center}.

The two main geometric tools we use are a generalization of a kernel of a simple polygon (Section~\ref{ssec:kernel}) and the weak visibility polygon for a line segment (Section~\ref{ssec:witness}). Finally, the existential proof can be turned into an efficient algorithm: the generalized kernel can be computed in $O(n\log n)$ time, and the visibility polygon for a point moving along a line segment can be maintained with a persistent data structure. The combination of these methods helps finding a witness point $s\in {\rm int}(P)$ with ${\rm int}(P)\subseteq V_{\floor{(n-2)/4}}(s)$  in $O(n \log n)$ time.

\section{Preliminaries}
\label{sec:prelim}

For a set $U\subset \mathbb{R}^2$ in the plane, let ${\rm int}(U)$ denote the interior, $\partial U$ the boundary, and ${\rm cl}(U)$ the closure of $U$. Let $P$ be a simply connected closed polygonal domain (for short, \emph{simple polygon}) with $n$ vertices. A \emph{chord} of $P$ is a closed line segment $ab$ such that $a,b\in \partial P$ and the relative interior of $ab$ is in ${\rm int}(P)$.

We assume that the vertices of $P$ are in general position (that is, no three collinear vertices), and we only consider light sources $s\in {\rm int}(P)$ that do not lie on any line spanned by two vertices of $P$. Recall that $V_0(s)$ is the visibility polygon of the point $s\in {\rm int}(P)$ with respect to $P$.
The \emph{pockets} of $V_0(s)$ are the connected components of $P\setminus {\rm cl}(V_0(s))$.
See Fig.~\ref{fig:pockets}(a) for examples.
The common boundary of $V_0(s)$ and a pocket is a chord $ab$ of $P$ (called a \emph{window}) such that $a$ is a reflex vertex of $P$ that lies in the relative interior of segment $sb$. We say that a pocket with a window $ab$ is \emph{induced by} the reflex vertex $a$. Note that every reflex vertex induces at most one pocket of $V_0(s)$. We define the \emph{size} of a pocket as the number of vertices of $P$ on the boundary of the pocket. Since the pockets of $V_0(s)$ are pairwise disjoint, the sum of the sizes of the pockets is at most $n$, the number of vertices of $P$.

\begin{figure}[ht]
  \centering
  \includegraphics[width=.9\textwidth]{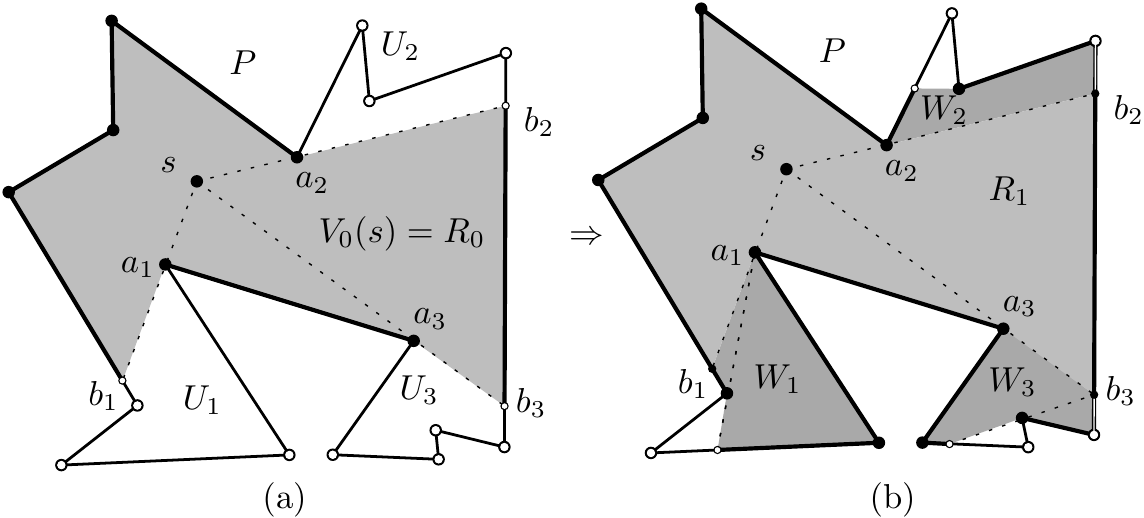}
  \caption{\label{fig:pockets}
(a) A polygon $P$ where $V_0(s)$ has three pockets $U_1$, $U_2$ and $U_3$,
    of size 4, 4, and 5, respectively. The left pockets are $U_1$ and $U_2$,
   the only right pocket is $U_3$. Pocket $U_1$ is independent of
   both $U_2$ and $U_3$; but $U_2$ and $U_3$ are dependent.
(b) The construction of region $R_1$ from $R_0=V_0(s)$ in~\cite{Us}.
Pocket $U_1$ is saturated, and pockets $U_2$ and $U_3$ are unsaturated.}
\end{figure}

A pocket is a \emph{left} (resp., \emph{right}) pocket if it lies on the left (resp., right) side of the directed line $\overrightarrow{ab}$. Two pockets of $V_0(s)$ are \emph{dependent} if some chord of $P$ crosses the window of both pockets; otherwise they are \emph{independent}.
One pocket is called independent if it is independent of all other pockets.

\begin{proposition} \label{prop:independent}
All left (resp., right) pockets of $V_0(s)$ are pairwise independent.
\end{proposition}

\begin{proof}
Consider two left pockets of $V_0(s)$, lying on the left side of the windows $\overrightarrow{a_1b_1}$ and $\overrightarrow{a_2b_2}$, respectively (see Fig.~\ref{fig:pockets}(a)). Suppose, for contradiction, that some chord $\ell$ of $P$ intersects both windows. Let $\ell'\subset \ell$ be the segment of $\ell$ between $a_1b_1$ and $a_2b_2$. Segment $\ell'$ lies in the right halfplane of both $\overrightarrow{a_1b_1}$ and $\overrightarrow{a_2b_2}$. The intersection of these two halfplanes is a wedge with the apex at $s$, and either $a_1b_1$ or $a_2b_2$ is not incident to this wedge.
This contradiction implies that no chord $\ell$ can cross both windows $a_1b_1$ and $a_2b_2$.
\end{proof}
The main result of Section~\ref{sec:prelim} is a sufficient condition (Lemma~\ref{lem:condition})
for a point $s\in {\rm int}(P)$ to fully illuminate ${\rm int}(P)$ within $\floor{(n-2)/4}$
diffuse reflections.
The proof of the lemma is postponed to the end of Section~\ref{sec:prelim}.
It relies on the techniques developed by Barequet et al.~\cite{Us} and the bound
$D(P)\leq \floor{n/2}-1$ on the diffuse reflection diameter.

\begin{lemma} \label{lem:condition}
We have ${\rm int}(P)\subseteq V_{\floor{(n-2)/4}}(s)$
for a point $s\in {\rm int}(P)$ if the pockets of $V_0(s)$
satisfy these conditions:
\begin{enumerate}\itemsep -2pt
  \item[$\mathbf{C}_1$] every pocket has size at most $\floor{n/2}-1$; and
  \item[$\mathbf{C}_2$] the sum of the sizes of any two dependent pockets is at most $\floor{n/2}-1$.
\end{enumerate}
\end{lemma}

\subsection{Review of regions $R_k$.}

We briefly review the necessary tools developed by Barequet et al.~\cite{Us}. Let $s \in {\rm int}(P)$ be a point in general position with respect to the vertices of $P$. Recall that $V_k(s)$, the set of points reachable from $s$ with at most $k$ diffuse reflections, is not necessarily simply connected when $k\geq 1$~\cite{BPS04}. Instead of tackling $V_k(s)$ directly, Barequet et al.~\cite{Us} recursively define simply connected regions $R_k=R_k(s)$, where $R_k(s)\subseteq V_k(s)$, for all $k\in \mathbb{N}_0$. For $k=0$, we have $R_0=V_0(s)$. We now review how $R_{k+1}$ is constructed from $R_k$. Each region $R_k$ is bounded by chords of $P$ and segments along the boundary $\partial P$. The connected components of $P\setminus {\rm cl}(R_k)$ are the \emph{pockets} of $R_k$. Each pocket $U_{ab}$ of $R_k$ is bounded by a chord $ab$ such that $a$ is a reflex vertex of $P$, $b$ is an interior point of an edge of $P$, and the two edges of $P$ incident to $a$ are on the same side of the line $ab$ (these properties are maintained recursively for $R_k$).

A pocket $U_{ab}$ of $R_k$ is \emph{saturated} if every chord of $P$ that crosses $ab$ has one endpoint in $R_k$ and the other endpoint in $U_{ab}$. Otherwise, $U_{ab}$ is \emph{unsaturated}. Recall that for a point $s'\in P$, $V_0(s')$ is the set of points in $P$ visible from $s'$. We also introduce an analogous notation for a line segment $pq\subseteq P$: let $V_0(pq)$ denote the set of points in $P$ visible from any point in $pq$.

For a given point $s\in {\rm int}(P)$, the regions $R_k$ are defined as follows (refer to Fig.~\ref{fig:pockets}(b)). Let $R_0=V_0(s)$. If ${\rm int}(P)\subseteq R_k$,
then let $R_{k+1}={\rm cl}(R_k)=P$. If ${\rm int}(P)\not\subseteq R_k$, then $R_k$ has at least one pocket.
For each pocket $U_{ab}$, we define a set $W_{ab}\subseteq U_{ab}$:
If $U_{ab}$ is saturated, then let $W_{ab}=V_0(ab)\cap U_{ab}$.
If $U_{ab}$ is unsaturated, then let $p_{ab}\in R_k\cap \partial P$ be a point close to $b$
such that no line determined by two vertices of $P$ separates $b$ and $p_{ab}$;
and then let $W_{ab}=V_0(p_{ab})\cap U_{ab}$. Let $R_{k+1}$ be the union of ${\rm cl}(R_k)$
and the sets $W_{ab}$ for all pockets $U_{ab}$ of $R_k$. Barequet et al.~\cite{Us} prove
that $R_k\subseteq V_k(s)$ for all $k\in \mathbb{N}_0$.

\begin{remark}\label{remark:1}{\rm
Note that when a pocket $U_{ab}$ is unsaturated, then $p_{ab}$ is an interior point of
some edge $e$ of $P$. Since light does not propagate along the edge $e$, the regions $W_{ab}$ and $R_{k+1}$ do not contain $e\cap U_{ab}$. Consequently, there is a fine difference between \emph{independent} and \emph{saturated} pockets. Every saturated pocket of $R_k$ is independent from all other pockets (by definition), but an independent pocket of $R_k$ is not necessarily saturated. In Fig.~\ref{fig:remark} (a), $U_1$ and $U_2$ are dependent
pockets of $R_0$; region $R_1$ covers the interior of $U_2$, but not its boundary, and it
has a pocket $U_5\subset U_1$. Even though $U_5$ is independent of all other pockets of $R_1(s)$,
it is unsaturated: a chord between $U_5$ and the uncovered part of $U_2$ crosses $a_5b_5$.
Since $s$ is in general position, this phenomenon does not occur for $k=0$, and
every independent pocket of $V_0(s)$ is saturated.}
\end{remark}

\begin{figure}[htp]
  \centering
  \includegraphics[width=.9\textwidth]{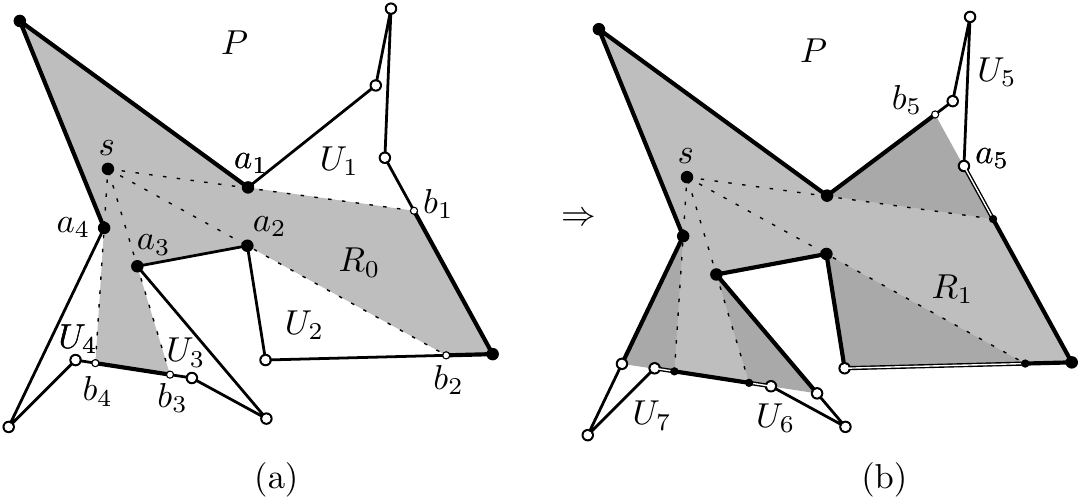}
  \caption{\label{fig:remark}
(a) A polygon $P$ where $R_0=V_0(s)$ has four unsaturated pockets: $U_1,\ldots , U_4$.
(b) The white lines on the boundary of $R_1$ are not part of $R_1$.
Consequently, pocket $U_5$ of $R_1$ is unsaturated, although it is independent of all other pockets.
Pockets $U_6$ and $U_7$ of $R_1$ are independent and saturated.}
\end{figure}

We say that a region $R_k$ \emph{weakly covers} an edge of $P$ if the boundary $\partial R_k$
intersects the relative interior of that edge. On the boundary of every pocket $U_{ab}$ of $R_k$,
there is an edge of $P$ that $R_k$ does not weakly cover, namely, the edge of $P$ incident to $a$.
We call this edge the \emph{lead edge} of $U_{ab}$. The following observation follows from
the way the regions $R_k$ are constructed in~\cite{Us}.

\begin{proposition}[\cite{Us}]\label{pp:Us}
For every pocket $U$ of region $R_k$, $k\in \mathbb{N}_0$,
the lead edge of $U$ is weakly covered by region $R_{k+1}$
and \emph{is not} weakly covered by $R_k$.
\end{proposition}

\begin{proposition}\label{pp:Us+}
If a pocket $U_{ab}$ of $V_0(s)$ has size $m$, then $R_k$
weakly covers at least $\min(k+1,m)$ edges of $P$ on the boundary of $U$.
\end{proposition}

\begin{proof}
  For every $k\in \mathbb{N}_0$, let $\mu_k$ denote the number of edges of $P$ on the boundary of $U_{ab}$ that are weakly covered by $R_k$.
$U_{ab}$ is bounded by a chord and $m$ edges of $P$.
One of these edges (the edge that contains $b$) is weakly covered by $V_0(s)$,
  hence $\mu_0\geq 1$.
Since $R_k\subseteq R_{k+1}$ for all $k\in \mathbb{N}_0$,
  $\mu_k$ is monotonically increasing, and every pocket of $R_k$
  that intersects $U$ is contained in $U_{ab}$.
In each pocket of $R_k$, by Proposition~\ref{pp:Us}, region $R_{k+1}$ weakly covers at least one new edge of $P$.
Consequently, we have $\mu_{k+1}\geq \min(\mu_k+1,m)$ for all $k\in \mathbb{N}_0$.
Induction on $k\in \mathbb{N}_0$ yields $\mu_k\geq \min(k+1,m)$.
\end{proof}

\subsection{Incrementally covering the pockets of $V_0(s)$}

In this subsection, we present three technical lemmas that yield upper
bounds on the minimum $k$ for which $V_k(s)$ contains the interior of
a given pocket of $V_0(s)$.
The following lemma is a direct consequence of Proposition~\ref{pp:Us+}.
It will be used for unsaturated pockets of $V_0(s)$.

\begin{lemma}\label{lem:allpockets}
If $U$ is a size-$m$ pocket of $V_0(s)$, then ${\rm int}(U)\subseteq R_{m-1}$.
\end{lemma}

\begin{proof}
By Proposition~\ref{pp:Us+}, $R_{m-1}$ weakly covers all edges of $P$ on the boundary of $U$.
Consequently, $U$ cannot contain any pocket of $R_{m-1}$ (otherwise $U\cap R_m$ would weakly
cover at least $m+1$ edges by Proposition~\ref{pp:Us}).
Thus ${\rm int}(U)\subseteq R_{m-1}$, as claimed.
\end{proof}

For saturated pockets, the diameter bound~\cite{Us} allows a better result.

\begin{lemma} \label{lem:saturated}
If $U$ is a size-$m$ saturated pocket of $R_k$,
then ${\rm int}(U)\subseteq R_{k+\floor{m/2}}$.
\end{lemma}

\begin{proof}
Let $ab$ be the window of $U$.
Since $a$ is a reflex vertex of $P$, it is a convex vertex of the pocket $U$.
Refer to Fig.~\ref{fig:pprime}. Since $U_{ab}$ is saturated, every chord that crosses $ab$
is part of a diffuse reflection path that starts at $s$
and enters the interior of $U_{ab}$ after at most $k$ reflections.

\begin{figure}[htp]
  \centering
  \includegraphics[width=.8\textwidth]{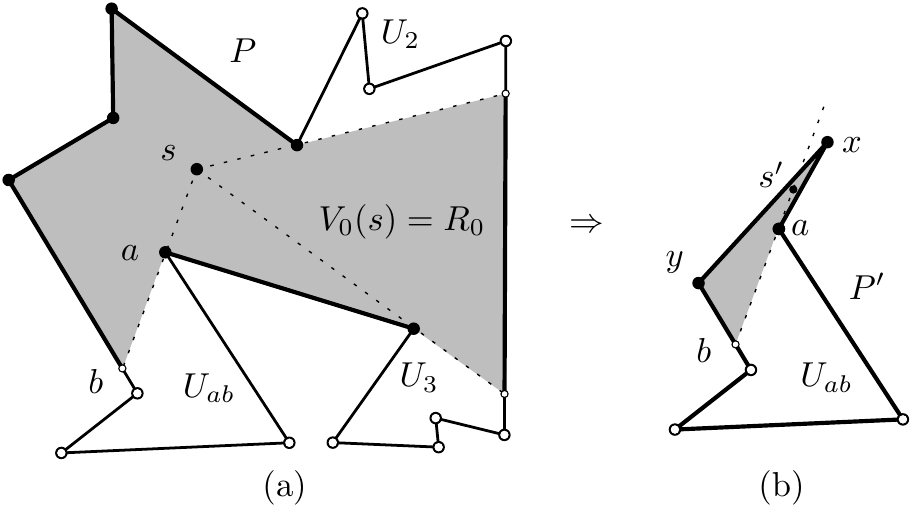}
  \caption{\label{fig:pprime}
(a) A polygon $P$ from Fig.~\ref{fig:pockets} with saturated pocket $U_{ab}$.
(b) Polygon $P'$ for the pocket $U_1$.}
\end{figure}

We construct a polygon $P'$ with $m+2$ vertices and a point $s'\in {\rm int}(P')$ such that $U$ is a pocket of $V_0(s')$ in $P'$,
and every chord of $P'$ that crosses $ab$ is part of a diffuse reflection path that starts at $s$ and enters the interior of $U$ after one reflection in $P'$.
The polygon $P'$ is bounded by the common boundary $\partial P\cap U$ and a polygonal path $(a,x,y,b)$, where $x$ is in a small neighborhood of $a$ such that $x$ and $U$ lie on the same side of line $ab$, and $y$ lies on the edge of $P$ that contains $b$ in the exterior of $U$. Place $s'\in {\rm int}(P')$ on the line $ab$ such that $a$ is in the relative interior of $s'b$.

Polygon $P'$ has $m+2$ vertices (since $b$ lies in the interior of an edge of $P$). The diffuse reflection diameter of a polygon with $m+2$ vertices is $\floor{(m+2)/2}-1 =\floor{m/2}$ from~\cite{Us}. Consequently, every point $t\in {\rm int}(U)$ can be reached from $s'$ after at most $\floor{m/2}$ diffuse reflections in $P'$.
Since a reflection path from $s'$ to any point $t\in {\rm int}(U)$ in $P'$ corresponds to an $s$-to-$t$ reflection path in the original polygon $P$
with at most $k$ more reflections, every $t\in {\rm int}(U)$ can be reached from $s$ after at most $k+\floor{m/2}$ diffuse reflections in $P$.
\end{proof}

Lemmas~\ref{lem:allpockets} and~\ref{lem:saturated} yield the following for dependent pockets of~$V_0(s)$.

\begin{lemma}\label{lem:dependent}
Let $U$ be a pocket of $V_0(s)$ of size $m$.
If each pocket dependent on $U$ has size at most $m'< m$,
  then ${\rm int}(U)\subseteq R_{\floor{(m+m')/2}}$.
\end{lemma}

\begin{proof}
For every $k\in \mathbb{N}_0$, let $\mu_k$ denote the number of edges of $P$ on the boundary of $U$ that are weakly covered by $R_k$. We have $\mu_0=1$, and if $\mu_k=m$, then ${\rm int}(U)\subseteq R_k$.
By Proposition~\ref{pp:Us+}, $\mu_{m'}\geq m'+1$  (i.e., at most $m-m'-1$ more edges
have to be weakly covered).

By Lemma~\ref{lem:allpockets}, $R_{m'-1}$ contains the interior of all pockets of $V_0(s)$ that depend on $U$. Consequently, if $R_{m'-1}$ has only one pocket inside $U$, it must be independent (but not necessarily saturated, cf. Remark~\ref{remark:1}). By definition, ${\rm cl}(R_{m'-1})\subseteq R_{m'}$, and so $R_{m'}$ also contains the boundaries of all pockets of $V_0(s)$ that depend on $U$. Consequently, if $R_{m'}$ has exactly one pocket inside $U$, it must be saturated.

We distinguish between two possibilities. First assume $\mu_{m'+k}\geq \min(m'+2k+1,m)$ for all $k\geq 1$ (that is, at least two more edges in $U$ get weakly covered until all edges in $U$ are exhausted). Then ${\rm int}(U)\subseteq R_{m'+\ceil{(m-m'-1)/2}} = R_{\floor{(m+m')/2}}$.

Otherwise, let $k\geq 1$ be the first index such that $\mu_{m'+k}=m'+2k<m$. Since $\mu_{m'+k-1}\geq m'+2(k-1)+1=m'+2k-1$ by assumption and $\mu_{m'+k} \geq  \mu_{m'+k-1}+1$ by Proposition~\ref{pp:Us}, we have $\mu_{m'+k} = \mu_{m'+k-1}+1$. This means that $R_{m'+k-1}$ has exactly one pocket in $U$, say $U_{ab}\subset U$, and $R_{m'+k+1}$ weakly covers only one new edge of $U_{ab}$ (e.g., pocket $U_5\subset U_1$ in Fig.~\ref{fig:remark}). This is possible only if $U_{ab}$ is unsaturated. Then the region $R_{m'+k}$ is extended by $W_{ab}=V_0(p_{ab})$ for a point $p_{ab}$ close to $b$. Since $W_{ab}$ weakly covers only one new edge, the lead edge of $U_{ab}$, which incident to $a$. Therefore, $W_{ab}$ is a triangle bounded by $ab$, the lead edge of $U_{ab}$, and the edge the contains $b$. It follows that $R_{m'+k}$ also has exactly one pocket in $U$, say $U_{a'b'}$, where the window $a'b'$ is collinear with the edge of $P$ that contains $b$. Hence the pocket $U_{a'b'}$ is \emph{saturated}:
of every chord that crosses $a'b'$, one endpoint is either in $W_{ab}\subset R_{m'+k+1}$ or in ${\rm cl}(R_{m'+k})\subset R_{m'+k+1}$. By Lemma~\ref{lem:saturated}, the interior of this pocket is contained in $R_{m'+k+ \ceil{(m-m'-2k-1)/2}} =R_{\floor{(m+m')/2}}$, as claimed.
\end{proof}

\subsection{Proof of Lemma~\ref{lem:condition}}
We prove a slightly more general statement than Lemma~\ref{lem:condition}.

\begin{lemma}\label{lem:condition+}
We have ${\rm int}(P)\subseteq V_k(s)$ if the pockets of $V_0(s)$ satisfy these conditions:
\begin{enumerate}\itemsep -2pt
    \item every pocket has size at most $2k+1$; and
    \item the sum of the sizes of any two dependent pockets is at most $2k+1$.
\end{enumerate}
\end{lemma}
\begin{proof}
Consider the pockets of $V_0(s)$. By Lemma~\ref{lem:allpockets},
the interior of every pocket of size at most $k+1$ is contained in $R_k$.
It remains to consider the pockets $U$ of size $m$ for $k+2\leq m\leq 2k+1$.
We distinguish between two cases.

\noindent{\bf Case~1: a pocket $U$ of size $m$ is independent of all
other pockets of $V_0(s)$.} Then $U$ is saturated (cf. Remark~\ref{remark:1}).
By Lemma~\ref{lem:saturated}, the interior of $U$ is
contained in $R_{\floor{m/2}}\subseteq R_k\subseteq V_k(s)$.

\noindent{\bf Case~2: a pocket $U$ of size $m$ is dependent on some other
pockets of $V_0(s)$.}
Any other pocket dependent on $U$ has size at most
  $m'=2k+1-m \leq k-1 < m$ by our assumption.
Lemma~\ref{lem:dependent} implies that
  the interior of $U$ is contained in $R_{\floor{(m+(2k+1-m))/2}}=R_k\subseteq V_k(s)$.
\end{proof}

\begin{proof}[of Lemma~\ref{lem:condition}]
Invoke Lemma~\ref{lem:condition+} with $k=\floor{(n-2)/4}$, and note that $2k+1=2\floor{(n-2)/4}+1\geq \floor{n/2}-1$.
\end{proof}

\subsection{Double Violators}
\label{ssec:double}

Recall that the sum of sizes of the pockets of $V_0(s)$ is at most $n$, the number of vertices of $P$.
It is, therefore, possible that several pockets or dependent pairs of pockets violate
conditions $\mathbf{C}_1$ or $\mathbf{C}_2$ in Lemma~\ref{lem:condition}.
We say that a point $s\in {\rm int}(P)$ is a
\emph{double violator} if $V_0(s)$ has either (i) two disjoint pairs of dependent pockets, each
pair with total size at least $\floor{n/2}$, or (ii) a pair of dependent pockets of total size at
least $\floor{n/2}$ and an independent pocket of size at least $\floor{n/2}$. (We do not worry
about the possibility of two independent pockets, each of size at least $\floor{n/2}$.)
In this section, we show that if there is a double violator
$s\in {\rm int}(P)$, then there is a point $s'\in {\rm int}(P)$ (possibly $s'=s$) for which
${\rm int}(P)\subseteq V_{\floor{(n-2)/4}}(s')$, and such an $s'$ can be found in $O(n)$ time.

The key technical tool is the following variant of Lemma~\ref{lem:dependent} for
a pair of dependent pockets that are adjacent to a common edge (that is, \emph{share} an edge).

\begin{lemma}\label{lem:double}
Let $U_{ab}$ and $U_{a'b'}$ be two dependent pockets of $V_0(s)$ such that neither is dependent on any other pocket, and points $b$ and $b'$ lie in the same edge of $P$. Let the size of $U_{ab}$ be $m$ and $U_{a'b'}$ be $m'$.
Then $R_{\floor{(m+m'-1)/2}}$ contains the interior of both $U_{ab}$ and $U_{a'b'}$.
\end{lemma}

\begin{proof}
For every $k\in \mathbb{N}_0$, let $\mu_k$ (resp., $\mu_k'$) denote the number of edges of $P$ on the boundary of $U_{ab}$ (resp., $U_{a'b'}$) that are weakly covered by $R_k$. We have $\mu_0=1$ and $\mu_0'=1$ (the edge containing $b$ and $b'$ is weakly covered by $V_0(s)$). Proposition~\ref{pp:Us} guarantees $\mu_1+\mu_1'\geq 4$. If $\mu_1+\mu_1'\geq 5$, then the proof of Lemma~\ref{lem:dependent} readily implies that $R_{\floor{(m+m'-1)/2}}$ contains the interior of both $U_{ab}$ and $U_{a'b'}$.

Assume now that $\mu_1+\mu_1'=4$. This means that $R_1$ weakly covers precisely one new edge from each of $U_{ab}$ and $U_{a'b'}$. Recall that $U_{ab}$ and $U_{a'b'}$ are unsaturated, and $R_1$ covers the part of $U_{ab}$ (resp., $U_{a'b'}$) visible from a point near $b$ (resp., $b'$). See Fig.~\ref{fig:double1}(a). It follows that $R_1$
has exactly one pocket in each of $U_{ab}$ and $U_{a'b'}$, and both pockets are on the same
side of the line $bb'$. Hence these pockets are saturated. They have size $m-1$ and $m'-1$, respectively.
By Lemma~\ref{lem:saturated}, the interiors of both $U_{ab}$ and $U_{a'b'}$ are covered by $R_{1+\floor{\max(m-1,m'-1)/2}}= R_{\floor{(\max(m,m')+1)/2}}\subseteq R_{\floor{(m+m'-1)/2}}$.
\end{proof}

\begin{figure}[htp]
  \centering
  \includegraphics[width=.95\textwidth]{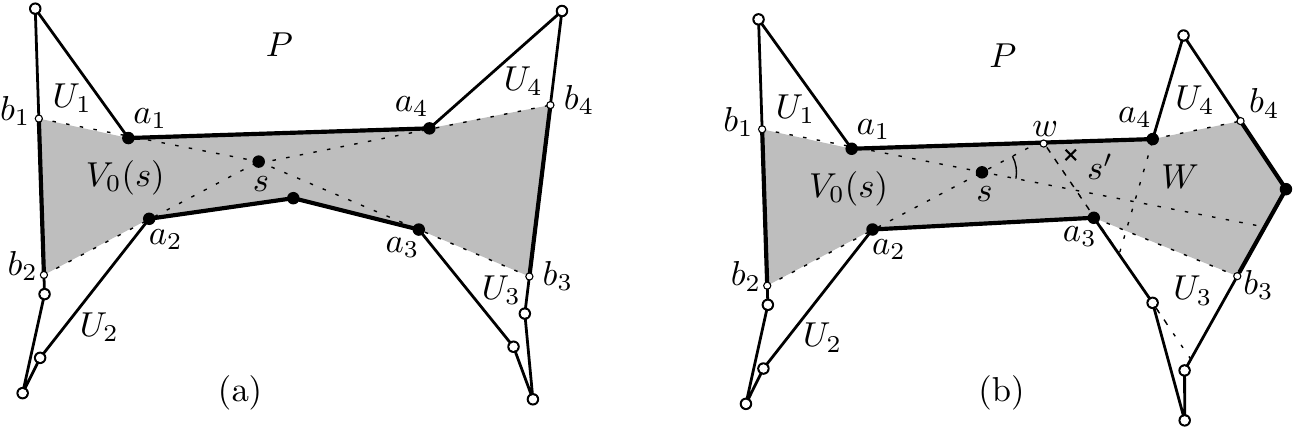}
  \caption{\label{fig:double1}
A polygon $P$ with $n=13$ vertices where $V_0(s)$ has four pockets: two pairs of dependent pockets, the
sum of sizes of each pair is $\floor{n/2} =6$.
(a) One extra vertex lies on $\partial P$ between two independent pockets.
(b) One extra vertex lies on $\partial P$ between two dependent pockets.}
\end{figure}

\begin{lemma}\label{lem:violate1}
Suppose that $V_0(s)$ has two disjoint pairs of dependent pockets, each
pair of total size at least $\floor{n/2}$. Then there is a point $s'\in {\rm int}(P)$ such that
${\rm int}(P)\subseteq V_{\floor{(n-2)/4}}(s')$, and such a point $s'$ can be computed in $O(n)$ time.
\end{lemma}

\begin{proof}
The sum of the sizes of these four pockets is at least $2\floor{n/2}$. If $n$ is even, then the two dependent pairs each have size $n/2$, they use all $n$ vertices of $P$, and both dependent pairs share an edge. If $n$ is odd, then either (i) the two dependent pairs have sizes $\floor{n/2}$ and $\ceil{n/2}$, resp., using all $n$ vertices of $P$, and both dependent pairs share an edge; or (ii) the two dependent pairs each have size $\floor{n/2}$, leaving one extra vertex, which may lie on the boundary between two independent pockets (Fig.~\ref{fig:double1}(a)), or between two dependent pockets (Fig.~\ref{fig:double1}(b)). In all cases, there is at least one dependent pair with joint size $\floor{n/2}$ that share an edge.

If the two dependent pairs each have size $\floor{n/2}$ and each share an edge (Fig.~\ref{fig:double1}(a)), then their interiors are covered by $R_k$ for $k=\floor{(\floor{n/2}-1)/2} =\floor{(n-2)/4}$ by Lemma~\ref{lem:double}.
This completely resolves that case that $n$ is even.

Assume now that $n$ is odd. Denote the four pockets by $U_1,\ldots , U_4$, induced by the reflex vertices $a_1,\ldots , a_4$ in counterclockwise order along $\partial P$, such that $U_1$ and $U_2$ are dependent with joint size $\floor{n/2}$ and share an edge; and $U_3$ and $U_4$ are dependent but either has joint size $\ceil{n/2}$ or do not share any edge. Refer to Fig.~\ref{fig:double1}(b). Note that $a_2a_3$ and $a_4a_1$ are edges of $P$. Let $W$ be the wedge bounded by the rays $\overrightarrow{a_1s}$ and $\overrightarrow{a_2s}$ (and disjoint from both $a_1$ and $a_2$). For every point $s'\in {\rm int}(P)\cap W$ in this wedge, $a_1$ and $a_2$ induce pockets $U_1'$ and $U_2'$, respectively, such that $U_1\subseteq U_1'$ and $U_2\subseteq U_2'$, and they also share an edge. Compute the intersection of region ${\rm int}(P)\cap W$ with the two lines containing the lead edges of $U_3$ and $U_4$.
Let $w$ be a closest point to $s$ on these segments, and let $s'\in {\rm int}(P)\cap W$ be a point close to $w$ in general position such that it can see all of the lead edge for $U_3$ or $U_4$. By construction, vertex $a_3$ or $a_4$ is not incident to any pocket of $V_0(s')$. Consequently, the total size of all pockets of $V_0(s')$ in $U_3$ and $U_4$ is at most $\floor{n/2}-1$. By Lemmas~\ref{lem:condition} and~\ref{lem:double}, $V_{\floor{(n-2)/4}}(s')$ contains the interiors
of all pockets of $V_0(s')$, as claimed.
\end{proof}

\begin{figure}[htp]
  \centering
  \includegraphics[width=.85\textwidth]{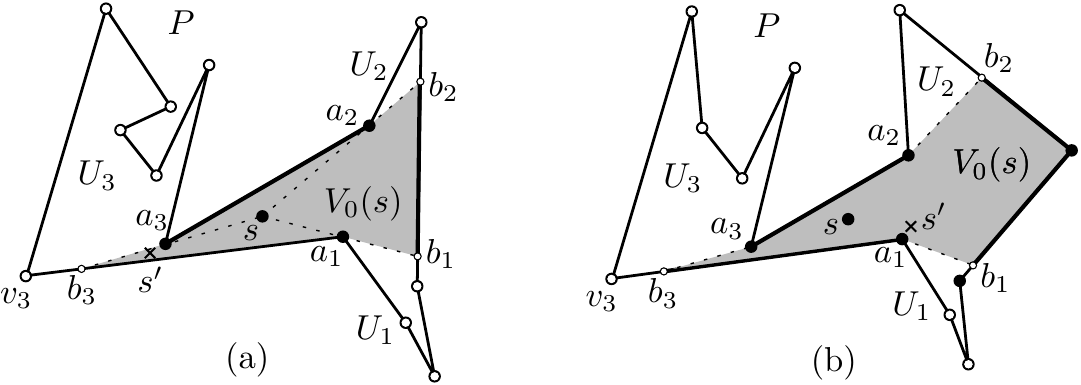}
  \caption{\label{fig:double2}
A polygon $P$ with $n=13$ vertices where $V_0(s)$ has three pockets: two dependent pockets of total size $\floor{n/2} =6$ and an independent pocket of size $\floor{n/2}=6$.
(a) One extra vertex lies on $\partial P$ between two independent pockets.
(b) One extra vertex lies on $\partial P$ between two dependent pockets}
\end{figure}

\begin{lemma}\label{lem:violate2}
Suppose that $V_0(s)$ has a pair of dependent pockets of total size at least $\floor{n/2}$ and
an independent pocket of size at least $\floor{n/2}$. Then there is a point $s'\in {\rm int}(P)$ with
${\rm int}(P)\subseteq V_{\floor{(n-2)/4}}(s')$, and $s'$ can be computed in $O(n)$ time.
\end{lemma}

\begin{proof}
The sum of the sizes of these three pockets is at least $2\floor{n/2}$.
This implies that $V_0(s)$ has no other pocket, and so the independent pocket is saturated (cf. Remark~\ref{remark:1}). If $n$ is even, then the two dependent pockets have total size $n/2$ and share an edge, and the independent pocket has size $n/2$.
If $n$ is odd, then either (i) the three pockets use all $n$ vertices of $P$, and the dependent pockets share an edge (Fig.~\ref{fig:double2}(a)); or (ii) the dependent pair and the independent pocket each have size $\floor{n/2}$, leaving one extra vertex, which may lie on the boundary between two independent pockets, or between two dependent pockets (Fig.~\ref{fig:double2}(b)).
Denote the three pockets by $U_1$, $U_2$, and $U_3$, induced by the reflex vertices $a_1$, $a_2$, and $a_3$ in counterclockwise order along $\partial P$, such that $U_1$ and $U_2$ are dependent; and $U_3$ is independent.
By Proposition~\ref{prop:independent}, $U_1$ and $U_2$ have opposite orientation, so we may assume without loss of generality that $U_2$ and $U_3$ have opposite orientation (say, left and right).

First suppose that $U_3$ has size $\ceil{n/2}$. Refer to Fig.~\ref{fig:double2}(a). Then $U_1$ and $U_2$ have joint size $\floor{n/2}$ and share an edge, and by Lemma~\ref{lem:double}, $R_{\floor{(n-2)/4}}$ contains the interior of both $U_1$ and $U_2$. Since all $n$ vertices are incident to pockets, $a_2a_3$ is an edge of $P$, and $a_1b_3$ is contained in an edge of $P$, say $a_1b_2 \subset a_1v_3$. Since $a_3b_3$ is a window of $V_0(s)$, the supporting line of $a_2a_3$ intersects segment $a_1b_3$. Let $s'\in {\rm int}(p)$ be a point close to the intersection of line $a_2a_3$ and segment $a_1b_3$. Then $a_1$ and $a_2$ induce pockets $U_1'$ and $U_2'$, respectively, such that $U_1\subseteq U_1'$, $U_2\subseteq U_2'$, and they share an edge. Both vertex $v_3$ and the lead edge of $U_3$ are directly visible from $s'$, they are not part of any pocket of $V_0(s')$. Consequently, the total size of pockets of $V_0(s')$ inside $U_3$ is at most $\ceil{n/2}-2$. By Lemma~\ref{lem:condition}, $V_{\floor{(n-2)/4}}(s')$ contains the interiors of all pocket of $V_0(s')$.

Now suppose that $n$ is odd and $U_3$ has size $\floor{n/2}$. Refer to Fig.~\ref{fig:double2}(b).
Denote the edge of $P$ that contains $b_3$ by $u_3v_3$ such that $v_3\in \partial U_3$ (and possibly $u_3=a_1$).
Let $s'\in {\rm int}(P)$ be a point in a small neighborhood of $u_3$. Then $s'$ directly sees $v_3$,
and similarly to the previous case, $V_{\floor{(n-2)/4}}(s')$ contains the interior of all pockets of $V_0(s')$ inside $U_3$. If $U_1$ and $U_2$ jointly have size $\ceil{n/2}=\floor{n/2}+1$, then they share an edge and $u_3=a_1$.
In this case $s'$ can see the lead edge of $U_1$, the total size of all pockets of $V_0(s')$ inside $U_1$ and $U_2$ is at most $\floor{n/2}$, and if it equals $\floor{n/2}$, then two of those pockets are dependent and share an edge. If $U_1$ and $U_2$ jointly have size $\floor{n/2}$, then $P$ has one ``unaffiliated'' vertex that does not belong to any pocket of $V_0(s')$ (Fig.~\ref{fig:double2}(b)). If $u_3 = a_1$, then $s'$ can see the lead edge of $U_1$, and thus the total size of all pockets of $V_0(s')$ inside $U_1$ and $U_2$ is at most $\floor{n/2}-1$. If $u_3\neq a_1$, then the unaffiliated vertex is $u_3$, hence $U_1$ and $U_2$ share an edge. Consequently, the total size of all pockets of $V_0(s')$ inside $U_1$ and $U_2$ is at most $\floor{n/2}$, and if it equals $\floor{n/2}$, then two of those pockets are dependent and share an edge. By Lemmas~\ref{lem:condition} and~\ref{lem:double}, ${\rm int}(P)\subseteq V_{\floor{(n-2)/4}}(s')$.
\end{proof}

\section{Finding a Witness Point}
\label{sec:center}

In Section~\ref{ssec:kernel}, we show that in every simple polygon $P$ in general position, there is a point $s\in {\rm int}(P)$ that satisfies condition $\mathbf{C}_1$.
In Section~\ref{ssec:witness}, we pick a point $s\in {\rm int}(P)$ that satisfies condition $\mathbf{C}_1$, and move it continuously until either (i) it satisfies both conditions $\mathbf{C}_1$ and $\mathbf{C}_2$, or (ii) it becomes a double violator. In both cases, we find a witness point for Theorem~\ref{thm:radius} (by Lemmas~\ref{lem:condition},~\ref{lem:violate1}, and~\ref{lem:violate2}).

\subsection{Generalized Kernel}
\label{ssec:kernel}

Let $P$ be a simple polygon with $n$ vertices. Recall that the set of points from which the entire polygon $P$ is visible is the \emph{kernel} of $P$, denoted $K(P)$, which is the intersection of all halfplanes bounded by a supporting line of an edge of $P$ and facing towards the interior of $P$. Lee and Preparata~\cite{LeePreparata79} designed an optimal $O(n)$ time algorithm for computing the kernel of simple polygon with $n$ vertices. We now define a generalization of the kernel. For an integer $q\in \mathbb{N}_0$, let $K_q(P)$ denote the set of points $s\in {\rm int}(P)$
such that every pocket of $V_0(s)$ has size at most $q$. Clearly, $K(P)=K_0(P)=K_1(P)$, and $K_q(P)\subseteq K_{q+1}(P)$ for all $q\in \mathbb{N}_0$. The set of points that satisfy condition $\mathbf{C}_1$ is $K_{\floor{n/2}-1}(P)$.

\begin{figure}[htp]
  \centering
  \includegraphics[width=\textwidth]{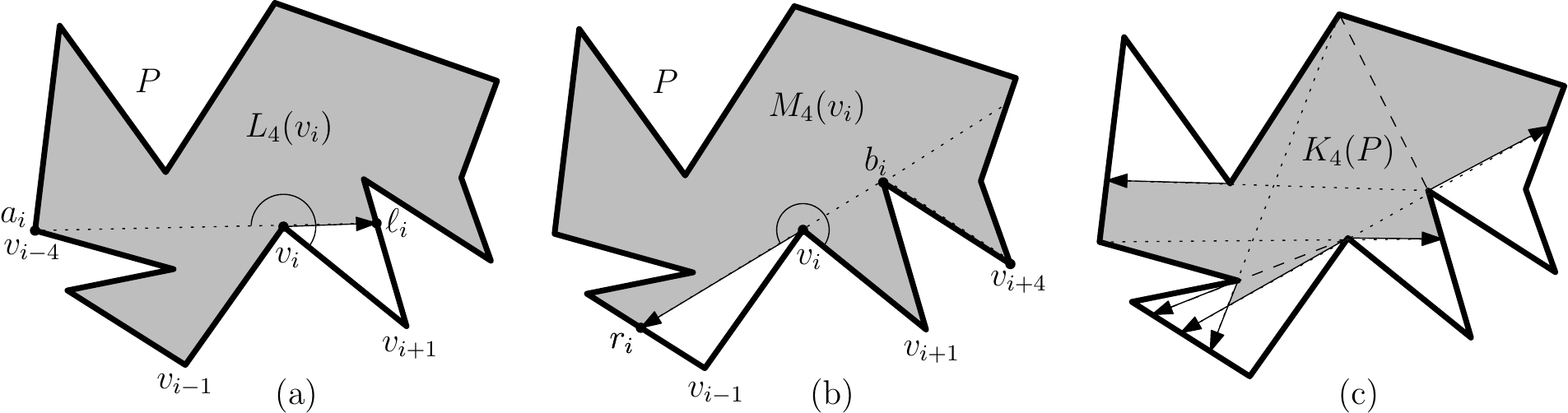}
  \caption{\label{fig:kernel}
(a) Polygon $L_4(v_i)$.
(b) Polygon $M_4(v_i)$.
(c) Polygon $K_4(P)$.}
\end{figure}

For every reflex vertex $v$, we define two polygons $L_q(v)\subset P$ and $M_q(v)\subset P$:
let $L_q(v)$ (resp. $M_q(v)$) be the set of points $s\in P$ such that $v$ does not induce a left (resp., right) pocket of size more than $q$ in $V_0(s)$. We have

$$K_q(P) = \bigcap_{v \mbox{\ \tiny\rm reflex}} \left(L_q(v) \cap M_q(v) \right).$$

We show how to compute the polygons $L_q(v)$ and $M_q(v)$. Refer to Fig.~\ref{fig:kernel}. Denote the vertices of $P$ by $(v_0,v_1,\ldots , v_{n-1})$, and use arithmetic modulo $n$ on the indices.
For a reflex vertex $v_i$, let $v_ia_i$ be the first edge of the shortest (geodesic) path from $v_i$ to $v_{i-q}$ in $P$. If the chord $v_ia_i$ and $v_iv_{i+1}$ meet at a reflex angle, then $v_ia_i$ is on the boundary of the \emph{smallest} left pocket of size at least $q$ induced by $v_i$ (for any source $s\in P$). In this case, the ray $\overrightarrow{a_iv_i}$ enters the interior of $P$, and we denote by $\ell_i$ the first point hit on $\partial P$. The polygon $L_q(v_i)$ is the part of $P$ lying on the left of the chord $\overrightarrow{v_i\ell_i}$. However, if the chord $v_ia_i$ and $v_iv_{i+1}$ meet at convex angle, then every left pocket induced by $v_i$ has size less than $q$, and we have $L_q(v_i)=P$.
Similarly, let $v_ib_i$ be the first edge of the shortest path from $v_i$ to $v_{i+q}$. Vertex $v_i$ can induce a right pocket of size more than $q$ only if $b_iv_i$ and $v_iv_{i-1}$ make a reflex angle. In this case, $v_ib_i$ is the boundary of the \emph{largest} right pocket of size at most $q$ induced by $v_i$, the ray $\overrightarrow{b_iv_i}$  enters the interior of $P$, and hits $\partial P$ at a point $m_i$, and $M_q(v_i)$ is the part of $P$ lying on the right of the chord $\overrightarrow{v_im_i}$. if $b_iv_i$ and $v_iv_{i-1}$ meet at a convex angle, then $M_q(v_i)=P$.

Note that every set $L_q(v_i)$ (resp., $M_q(v_i)$) is \emph{$P$-convex} (a.k.a.  \emph{geodesic convex}),
that is, $L_i(v_i)$ contains the shortest path between any two points in $L_q(v_i)$ with respect to $P$~\cite{BKOW14,DEH04,Tou86}. Since the intersection of $P$-convex polygons is $P$-convex, $K_q(P)$ is also $P$-convex for every $q\in \mathbb{N}_0$.
There exists a point $s\in {\rm int}(P)$ satisfying condition $\mathbf{C}_1$ if and only if $K_{\floor{n/2}-1}(P)$ is nonempty. We prove $K_{\floor{n/2}-1}(P)\neq \emptyset$ using a Helly-type result by Breen~\cite{Breen} (cf.~\cite{Breen98,Molnar57}).

\begin{theorem}[\cite{Breen}]\label{thm:Breen}
Let $\mathcal{P}$ be a family of simple polygons in the plane. If every three (not necessarily distinct)
members of $\mathcal{P}$ have a simply connected union and every two members of $\mathcal{P}$
have a nonempty intersection, then $\bigcap_{P \in \mathcal{P}}{P} \neq \emptyset$.
\end{theorem}
\begin{figure}[htp]
  \centering
  \includegraphics[width=.85\textwidth]{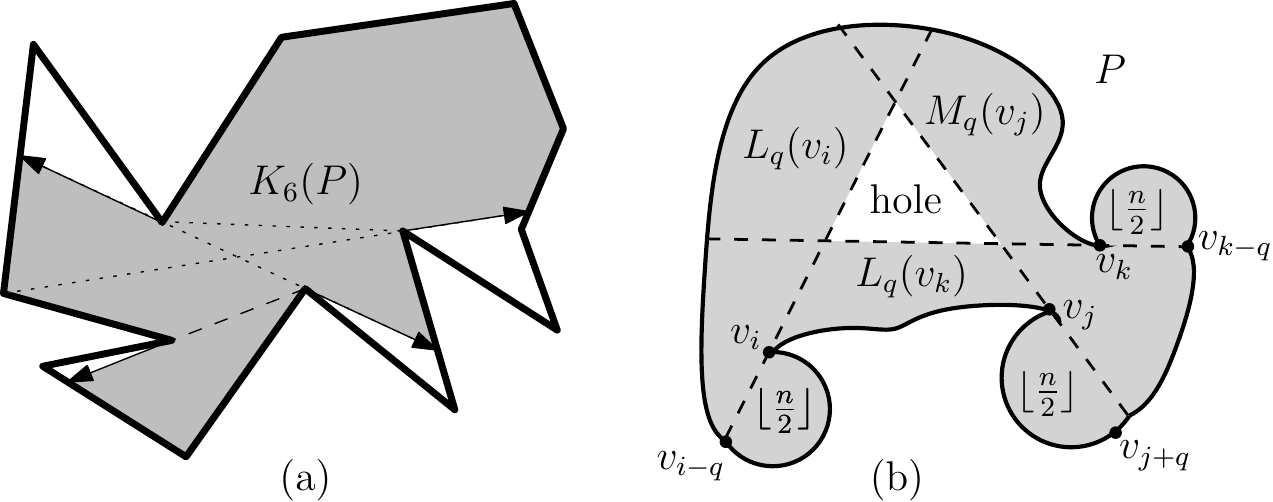}
  \caption{\label{fig:kernel2}
(a) A simple polygon $P$ with $n=14$ vertices, and the generalized kernel $K_{\floor{n/2}-1}(P)=K_6(P)$.
(b) A schematic picture of a triangular hole in the union of three polygons in $P$.}
\end{figure}
\begin{lemma}\label{lem:kernel}
  For every simple polygon $P$ with $n\geq 3$ vertices, $K_{\floor{n/2}-1}(P)$ has nonempty interior.
\end{lemma}
\begin{proof}
When $n=3$, we have $K_{\floor{n/2}-1}(P) = K_0(P) = P$.
Now assume $n > 3$.

We apply Theorem~\ref{thm:Breen} for the polygons $L_{\floor{n/2}-1}(v_i)$ and
$M_{\floor{n/2}-1}(v_i)$ for all reflex vertices $v_i$ of $P$. By definition,
$L_{\floor{n/2}-1}(v_i)$ is incident to $v_{i-\floor{n/2}+1},\ldots, v_i$ on or left of $\overrightarrow{v_ia_i}$, and similarly $M_{\floor{n/2}-1}(v_i)$) is incident to $v_i, \ldots, v_{i+\floor{n/2}-1}$ or or right or $\overrightarrow{v_ib_i}$.
Furthermore, $L_{\floor{n/2}-1}$ (resp., $M_{\floor{n/2}-1}$)
is incident to at least one additional vertex right of $\overrightarrow{v_ia_i}$
(resp., left of $\overrightarrow{v_ib_i}$). Thus, each of these sets is incident to
at least $\floor{n/2}+1$ vertices of $P$. Since $2(\floor{n/2}+1) > n$,
any two of these sets are incident to a common vertex of $P$ by the pigeonhole principle.

Recall that $L_i$ and $M_i$ are each bounded by part of the boundary of $P$
and possibly a chord incident to $v_i$. Consequently, if two of these sets are incident
to two or more common vertices of $P$ then their interiors intersect. If they are incident
to precisely one common vertex of $P$, then the common vertex, say $v_i$, is incident to
both boundary chords, hence the two sets are $L_i$ and $M_i$. In this case, however,
$v_i$ is a reflex vertex of both $L_i$ and $M_i$, and so their interiors intersect.

It remains to show that the union of any three of them is simply connected.
Suppose, to the contrary, that there are three sets whose union has a hole. Since each set is bounded by a chord of $P$, the hole must be a triangle bounded by the three chords on the boundary of the three polygons. Refer to Fig.~\ref{fig:kernel2}(b).
Each of these chords is incident to a reflex vertex of $P$ and is collinear with \emph{another} chord of $P$ that weakly separates the vertices
$\{v_i,v_{i+1}, \ldots, v_{i+\floor{n/2}-1}\}$ or
$\{v_i,v_{i-1}, \ldots, v_{i-\floor{n/2}+1}\}$ from the hole.
Figure~\ref{fig:kernel2}(b) shows a schematic image.
The latter three chords together weakly separate disjoint sets of vertices
  of total size at least $3\floor{n/2}>n$ from the hole,
contradicting the fact that $P$ has $n$ vertices altogether.
\end{proof}

By Lemma~\ref{lem:kernel}, $K_{\floor{n/2}-1}(P)$ has nonempty interior,
so there is a light source $s\in {\rm int}(K_{\floor{n/2}-1}(P))$ 
that satisfies condition $\mathbf{C}_1$.

\begin{lemma}\label{lem:kernel2}
For every $q\in \mathbb{N}_0$, $K_q(P)$ can be computed in $O(n\log n)$ time.
\end{lemma}

\begin{proof}
With a shortest path data structure~\cite{GH89} in a simple polygon $P$, the first edge of the shortest path
between any two query points can be computed in $O(\log n)$ time after $O(n)$ preprocessing time. A ray shooting data structure~\cite{HS95} can answer ray shooting queries in $O(\log n)$ time after $O(n)$ preprocessing time. Therefore, any chord $\overrightarrow{v_i\ell_i}$ or $\overrightarrow{v_im_i}$ can be computed in $O(\log n)$ time.

The generalized kernel $K_q(P) = \bigcap \left(L_q(v_i) \cap M_q(v_i) \right)$, can be constructed by incrementally maintaining the intersection $K$ of some sets from $\{L_q(v_i), M_q(v_i): v_i$ is reflex$\}$. In each step, we compute the intersection of $K$ with $L_q(v_i)$ or $M_q(v_i)$. Recall that all these sets are $P$-convex (the intersection of $P$-convex sets is $P$-convex). A chord of $P$ intersects the boundary of a $P$-convex polygon $K$ in at most two points, and the intersection points can be computed in $O(\log n)$ time using a ray-shooting query in $P$ (shoot a ray along the chord, and find the intersection points with binary search along the boundary of $K$). Thus $K$ can be updated in $O(\log n)$ time. Altogether, we can compute $K_q(P)$ in $O(n\log n)$ time.
\end{proof}

\subsection{Finding a Witness}
\label{ssec:witness}

In this section, we present an algorithm that, given a simple polygon $P$ with $n$ vertices in general position, finds a witness $s\in {\rm int}(P)$ such that ${\rm int}(P)\subseteq V_{\floor{(n-2)/4}}(s)$.

Let $s_0$ be an arbitrary point in ${\rm int}(K_{\floor{n/2}-1}(P))$. Such a point exists by Lemma~\ref{lem:kernel}, and can be computed in $O(n \log n)$ time by Lemma~\ref{lem:kernel2}. We can compute the visibility polygon $V_0(s_0)$
and its pockets in $O(n)$ time~\cite{GHL+87}. The definition of $K_{\floor{n/2}-1}(P)$ ensures that $s_0$ satisfies condition $\mathbf{C}_1$ of Lemma~\ref{lem:condition}. If it also satisfies $\mathbf{C}_2$, then $s=s_0$ is a desired witness.

Assume that $s_0$ does not satisfy $\mathbf{C}_2$, that is, $V_0(s_0)$ has two dependent pockets of total size at least $\floor{n/2}$, say a left pocket $U_{ab}$ and (by Proposition~\ref{prop:independent}) a right pocket $U_{a'b'}$. We may assume that $U_{ab}$ is at least as large as $U_{a'b'}$, by applying a reflection if necessary, and so the size of $U_{ab}$ is at least $\floor{n/4}$. Refer to Fig.~\ref{fig:line}(a). Let $c\in \partial P$ be a point sufficiently close to $b$ such that segment $bc$ is disjoint from all lines spanned by the vertices of $P$, segment $s_0c$ is disjoint from the intersection of any two lines spanned by the vertices of $P$, and $s_0c\subset P$. In Lemma~\ref{lem:line} (below), we find a point on segment $s_0c$ that is a witness, or double violator, or improves a parameter (spread) that we introduce now.

For a pair of dependent pockets, a left pocket $U_{ab}$ and (by Proposition~\ref{prop:independent}) a right pocket $U_{a'b'}$, let $\spread(a,a')$ be the part of $\partial P$ clockwise from $a$ to $a'$ (inclusive), and let the \emph{size} of $\spread(a,a')$ be the number of vertices of $P$ along $\spread(a,a')$. Note that
$|\spread(a,a')|$ is at least the sum of the sizes of the two dependent pockets, as all vertices incident to the two pockets are counted. For a pair of pockets of total size at least $\floor{n/2}$, we have $\floor{n/2}\leq |\spread(a,a')|\leq n$.

\begin{figure}[htbp]
  \centering
  \includegraphics[width=0.8\textwidth]{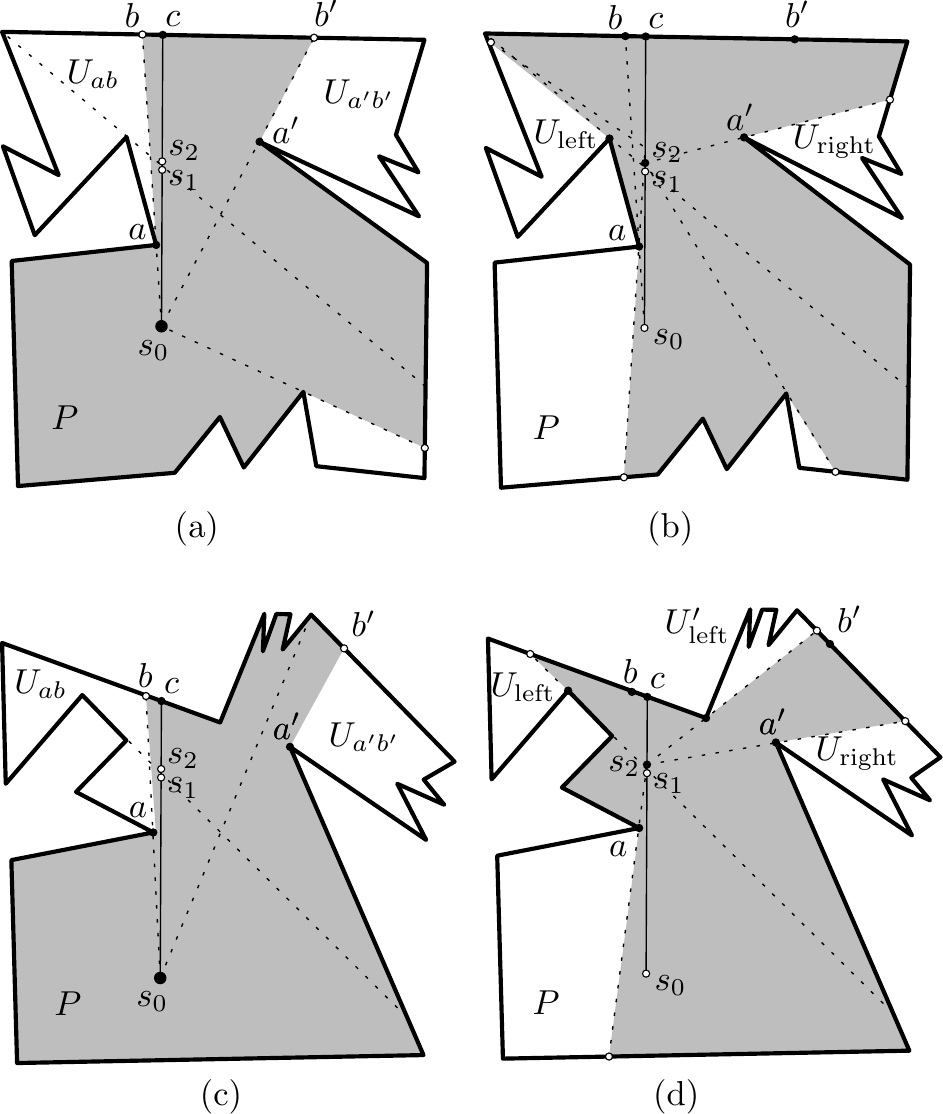}
  \caption{\label{fig:line}
(a) A polygon with $n=21$ vertices where $s_0$ violates $\mathbf{C}_2$ a pair of dependent
    pockets $U_{ab}$ and $U_{a'b'}$.
(b) Point $s_2\in s_0c$ satisfies both $\mathbf{C}_1$ and $\mathbf{C}_2$.
(c) A polygon with $n=21$ vertices where $s_0$ violates $\mathbf{C}_2$ with
    a pair of pockets $U_{ab}$ and $U_{a'b'}$ of $|\spread(a,a')|=19$.
(d) Point $s_2$ also violates $\mathbf{C}_2$ with a pair of pockets of $|\spread(a'',a')|=13$.}
\end{figure}

We introduce some terminology to trace effects of moving a point $s$ continuously in the interior of $P$.
The visibility polygons of two points are \emph{combinatorially equivalent} if there is a bijection between their pockets such that corresponding pockets are incident to the same sets of vertices of $P$.
The combinatorial changes incurred by a moving point $s$ have been thoroughly analysed in \cite{AGTZ02,BLM02,CW12}.
The set of points $s\in P$ that induces combinatorially equivalent visibility polygons $V_0(s)$ forms a \emph{cell} in the \emph{visibility decomposition} $VD(P)$ of polygon $P$. It is known that each cell is convex and there are $O(n^3)$ cells, but a line segment in $P$ intersects only $O(n)$ cells~\cite{BLM02,CD98}.
A combinatorial change in $V_0(s)$ occurs if $s$ crosses a \emph{critical line} spanned by two vertices of $P$, and the circular order of the rays from $s$ to the two vertices is reversed.
The possible changes are: (1) a pocket of size 2 appears or disappears; (2) the size of a pocket increases or decreases by one; (3) two pockets merge into one pocket or a pocket splits into two pockets.
Importantly, the combinatorics of $V_0(s)$ does not contain enough information to decide whether two pockets are dependent or independent. Proposition~\ref{prop:sliding-pockets} (below) will be crucial for checking whether two dependent pockets become independent when a point $s$ moves along a straight-line trajectory from $s=s_1$ to $s=s_2$.

\begin{proposition} \label{prop:sliding-pockets}
Let $s_1s_2$ be a line segment in ${\rm int}(P)$. Then
\begin{enumerate}[(i)]
\item Every left (resp., right) pocket of $V_0(s_2)$ induced by a
vertex on the left (right) of $\overrightarrow{s_1s_2}$
is contained in a left (right) pocket of $V_0(s_1)$.
\item Let $U_{\rm left}$ and $U_{\rm right}$ be independent pockets of $V_0(s_1)$.
Then every two pockets of $V_0(s_2)$ contained in $U_{\rm left}$ and $U_{\rm right}$, respectively, are also independent.
\end{enumerate}
\end{proposition}

\begin{proof}
(i) Let $U_{ab}$ be a left pocket of $V_0(s_2)$ induced by vertex $a$ on the left of $\overrightarrow{s_1s_2}$. If $a$ is directly visible from $s$ (i.e., $s_1a\subset P$),
then $U_{ab}$ is clearly contained in the left pocket of $V_0(s_1)$ induced by $a$.
Otherwise, consider the geodesic path from $s_1$ to $a$ in $P$. It is homotopic to the path $(s_1,s_2,a)\subseteq P$, and so it contained in the triangle~$\Delta(s_1,s_2,a)$. The first
internal vertex of this geodesic induces a left pocket of $V_0(s_1)$ that contains $U_{ab}$.

(ii) Since $U_{\rm left}$ and $U_{\rm right}$ are independent, no chord of $P$ crosses the window of both pockets. Therefore no chord of $P$ can cross the windows of two pockets lying in
$U_{\rm left}$ and $U_{\rm right}$, respectively.
\end{proof}

\begin{lemma}\label{lem:line}
Let $s_0$ be an arbitrary point in ${\rm int}(K_{\floor{n/2}-1}(P))$ and
$c \in \partial P$ as defined above.
Then there is a point $s\in s_0c$ such that one of the following statements holds:
\begin{itemize}
\item[$\bullet$] $s$ satisfies both $\mathbf{C}_1$ and $\mathbf{C}_2$;
\item[$\bullet$] $s$ is a double violator;
\item[$\bullet$] $s$ satisfies $\mathbf{C}_1$ but violates $\mathbf{C}_2$ due
                  to two pockets of whose spread is contained in $\spread(a,a')$
                  and has size at most $|\spread(a,a')|-\floor{n/4}$.
\end{itemize}
\end{lemma}
\begin{proof}
We move a point $s\in s_0c$ from $s_0$ to $c$ and trace the combinatorial changes
of the pockets of $V_0(s)$, and their dependencies. Initially, when $s=s_0$, all pockets
have size at most $\floor{n/2}-1$; and there are two dependent pockets, a left pocket $U_{ab}$ on the left of $\overrightarrow{s_0c}$ and, by Proposition~\ref{prop:independent}, a right pocket $U_{a'b'}$ on the right of $\overrightarrow{s_0c}$, of total size at least $\floor{n/2}$. When $s=c$, every left pocket of $V_0(s)$ on the left of $\overrightarrow{s_0c}$ is independent of any right pocket on the right of $\overrightarrow{s_0c}$.

Consequently, when $s$ moves from $s_0$ to $c$, there is a critical change from $s=s_1$ to $s=s_2$ such that $V_0(s_1)$ still has two dependent pockets of size at least $\floor{n/2}$ where the left (resp., right) pocket is on the left (right) of $\overrightarrow{s_0c}$; but $V_0(s_2)$ has no two such pockets. (See Fig.~\ref{fig:line} for examples.) Let $U_{\rm left}$ and $U_{\rm right}$ denote the two violator pockets of $V_0(s_1)$. The critical point is either a combinatorial change (i.e., the size of one of these pockets drops), or the two pockets become independent. By Proposition~\ref{prop:sliding-pockets}, we have $U_{\rm left}\subseteq U_{ab}$ and $U_{\rm right}\subset P\setminus U_{ab}$, and the spread of $U_{\rm left}$ and $U_{\rm right}$ is contained in $\spread(a,a')$. We show that one of the statements in Lemma~\ref{lem:line} holds for $s_1$ or $s_2$.

If $s_2$ satisfies both $\mathbf{C}_1$ and $\mathbf{C}_2$, then our proof is complete (Fig.~\ref{fig:line}(a-b)).
If $s_2$ violates $\mathbf{C}_1$, i.e., $V_0(s_2)$ has a pocket of size $\geq \floor{n/2}$,
then $V_0(s_1)$ also has a combinatorially equivalent pocket (which is independent of $U_{\rm left}$ and $U_{\rm right}$), and so $s_1$ is a double violator. Finally, if $s_2$ violates $\mathbf{C}_2$, i.e., $V_0(s_2)$ has two dependent pockets of total size $\floor{n/2}$, then the left pocket of this pair is not contained in $U_{ab}$ by the choice of point $c\in \partial P$. We have two subcases to consider: (i) If the right pocket of this new pair is contained in $U_{\rm right}$ (or it is $U_{\rm right}$), then we know that their spread is contained in $\spread(b,a)$ which has size at most $|\spread(a,a')|-\floor{n/4}$  (Fig.~\ref{fig:line}(c-d)). (ii) If the right pocket of the new pair is disjoint from $U_{\rm right}$, then $V_0(s_1)$ also has a combinatorially equivalent pair of pockets, which is different from $U_{\rm left}$ and $U_{\rm right}$, and so $s_1$ is a double violator.
\end{proof}

\begin{lemma}\label{lem:line2}
A point $s\in s_0c$ described in Lemma~\ref{lem:line} can be found in $O(n \log n)$ time.
\end{lemma}
\begin{proof}
It is enough to show that the critical positions, $s_1$ and $s_2$, in the proof of Lemma~\ref{lem:line} can be computed in $O(n\log n)$ time. We use the persistent data structure developed
by Chen and Daescu~\cite{CD98} for maintaining the combinatorial structure of $V_0(s)$ as $s$ moves along the line segment $s_0c$. The pockets (and pocket sizes) change only at $O(n)$ points
along $s_0c$, and each update can be computed in $O(\log n)$ time.

However, the data structure in~\cite{CD98} does not track whether two pockets on opposite sides
of $s_0c$ are dependent or not. The main technical difficulty is that $\Omega(n^2)$ dependent pairs might become independent as $s$ moves along $s_0s$ (even if we consider only pairs of total  size at least $\floor{n/2}$), in contrast to only $O(n)$ combinatorial changes. We reduce the number of relevant events by focusing on only the ``large'' pockets (pockets of size at least $\floor{n/4}$), and maintaining at most one pair that violates $\mathbf{C}_2$ for each large pocket. (In a dependent pair of size $\geq \floor{n/2}$, one of the pockets has size $\geq \floor{n/4}$.)

We augment the persistent data structure in~\cite{CD98} as follows. We maintain the list of all left (resp., right) pockets of $V_0(s)$ lying on the left (right) of $\overrightarrow{s_0c}$, sorted in counterclockwise order along $\partial P$. We also maintain the set of \emph{large} pockets of size at least $\floor{n/4}$ from these two lists. There are at most 4 large pockets for any $s\in s_0c$.
For a large pocket $U_{\alpha\beta}$ of $s\in s_0c$, we maintain one possible other pocket $U_{\alpha'\beta'}$ of $V_0(s)$ such that they together violate $\mathbf{C}_2$. If there are several such pockets $U_{\alpha'\beta'}$, we maintain only the one where $\alpha'$ (the reflex vertex that induces $U_{\alpha'\beta'}$) is farthest from $c$ along $\partial P$. Thus, we maintain a set $\mathcal{U}(s)$ of at most 4 pairs $(U_{\alpha\beta}, U_{\alpha'\beta'})$. Finally, for each of pair $(U_{\alpha\beta}, U_{\alpha'\beta'})\in \mathcal{U}$, we maintain the positions $s'= sc\cap \alpha\alpha'$ where the pair $(U_{\alpha\beta}, U_{\alpha'\beta'})$ becomes independent assuming that neither $U_{\alpha\beta}$ nor $U_{\alpha'\beta'}$ goes through combinatorial changes before $s$ reaches $s'$. We use~\cite{CD98} together with these supplemental structures, to find critical points $s_1,s_2\in  s_0c$ such that $\mathcal{U}(s_1)\neq \emptyset$ but $\mathcal{U}(s_2)=\emptyset$.

We still need to show that $\mathcal{U}(s)$ can be maintained in $O(n\log n)$ time as $s$ moves from $s_0$ to $c$. A pair $(U_{\alpha\beta}, U_{\alpha'\beta'})$  has to be updated if $U_{\alpha\beta}$ or $U_{\alpha'\beta'}$ undergoes a combinatorial change, or if they become independent (i.e., $s\in \alpha\alpha'$). Each large pocket undergoes $O(n)$ combinatorial changes by Proposition~\ref{prop:sliding-pockets}. Note also that there are $O(n)$ reflex vertices along the boundary $\partial P$ between $a$ and $a'$ (these vertices are candidates to become $\alpha'$).
No update is necessary when $\beta$ or $\beta'$ changes but $U_{\alpha\beta}$ remains large and the total size of the pair is at least $\floor{n/2}$. If the size of $U_{\alpha\beta}$ drops below $\floor{n/4}$, we can permanently eliminate the pair from $\mathcal{U}$. In all other cases, we search for a new vertex $\alpha'$, by testing the reflex vertices that induce pockets from the current $\alpha'$ towards $c$ along $\partial P$ until we either find a new pocket $U_{\alpha'\beta'}$ or determine that $U_{\alpha\beta}$ is not dependent of any other pocket with joint size $\geq \floor{n/2}$. We can test dependence between $U_{\alpha\beta}$ and a candidate for $U_{\alpha'\beta'}$ in $O(\log n)$ time (test $\alpha\alpha'\subset P$ by a ray shooting query). Each update of $(U_{\alpha\beta}, U_{\alpha'\beta'})$ decreases the size of the large pocket $U_{\alpha\beta}$ or moves the vertex $\alpha'$ closer to $c$. Therefore, we need to test dependence between only $O(n)$ candidate pairs of pockets. Overall, the updates to $\mathcal{U}(s)$ take $O(n \log n)$ time.
\end{proof}

We are now ready to prove Theorem~\ref{thm:radius}.
\begin{proof}[of Theorem~\ref{thm:radius}]
Let $P$ be a simple polygon with $n \geq 3$ vertices.
Compute the generalized kernel $K_{\floor{n/2}-1}(P)$,
and pick an arbitrary point $s_0\in {\rm int}(K_{\floor{n/2}-1}(P))$,
which satisfies $\mathbf{C}_1$. If $s_0$ satisfies $\mathbf{C}_2$, too,
then ${\rm int}(P)\subseteq V_\floor{(n-2)/4}(s_0)$ by Lemma~\ref{lem:condition}.
Otherwise, there is a pair of dependent pockets, $U_{ab}$ and $U_{a'b'}$,
of total size at least $\floor{n/2}$ and $\floor{n/2}\leq {\rm spread}(a,a')\leq n$.
Invoke Lemma~\ref{lem:line} up to three times to find a point $s\in {\rm int}(P)$
that either satisfies both $\mathbf{C}_1$ and $\mathbf{C}_2$, or is a double violator.
If $s$ satisfies $\mathbf{C}_1$ and $\mathbf{C}_2$ then Lemma~\ref{lem:condition}
completes the proof. If $s$ is a double violator, apply Lemma~\ref{lem:violate1} or Lemma~\ref{lem:violate2} as appropriate to complete the proof. The overall running time
of the algorithm is $O(n\log n)$ from the combination of Lemmas~\ref{lem:violate1},~\ref{lem:violate2},~\ref{lem:kernel2}, and~\ref{lem:line2}.

For every $k\geq 1$, the diffuse reflection diameter of the zig-zag polygon (cf. Fig.~\ref{fig:diffuse-ex1}) with $n=4k+2$ vertices is $k=\floor{(n-2)/4}$. By introducing up to 3 dummy vertices on the boundary of a zig-zag polygon, we obtain $n$-vertex polygons $P_n$ with $R(P_n)=\floor{(n-2)/4}$ for all $n\geq 6$. Finally, every simple polygon with $n=3$, 4, or 5 vertices is star-shaped, and so its diffuse reflection radius is $0=\floor{(n-2)/4}$.
\end{proof}

\section{Approximate Diffuse Reflection Radius}

In this section, we prove Theorem~\ref{thm:apx-compute-radius} and show how to \emph{approximate} the diffuse reflection radius $R(P)$ of a given polygon $P$ up to an additive error of at most 1 in polynomial time (a similar strategy works for approximating the diffuse reflection diameter $D(P)$, as well).

\begin{proof}[of Theorem~\ref{thm:apx-compute-radius}]
Let $P$ be a simple polygon with vertex set $V$ in general position, where $|V| = n$.
We wish to compute an integer $k\in \mathbb{N}$ such that $k-1\leq R(P)\leq k+1$,
and a point $s\in {\rm int}(P)$ such that ${\rm int}(P)\subseteq V_{k+1}(s)$ in polynomial time.
We prove the claim by analyzing the following algorithm:
\begin{enumerate} \itemsep2pt
\item[]\texttt{ApproxDiffuseRadius}$(P)$
  \item For each vertex $v$ of $P$, find two points $v^-$ and $v^+$ in the relative interior of the two edges of $P$ incident to $v$ such that no line through a pair of vertices in
       $V\setminus \{v\}$ separates them from $v$. Let $Q=\{v^-,v^+:v\in V\}$.
  \item Find the minimum integer $k\geq 0$ such that $C_k=\bigcap_{q\in Q} V_k(q)\cap {\rm int}(P)$
       is nonempty by binary search over $k \in \{0, \dots, \lfloor (n-2)/4 \rfloor\}$.
\item Return $k$ and an arbitrary point $s\in C_k$.
\end{enumerate}

We first show that \texttt{ApproxDiffuseRadius}$(P)$ runs in polynomial time in~$n$.
We can find a suitable set $Q=\{v^-,v^+:v\in V\}$ in $O(n^3)$ time by computing, for each $v\in V$,
the intersection points between the $O(n^2)$ lines through a pair of vertices in $V\setminus \{v\}$ and
the two edges of $P$ incident to $v$.
Then, $v^-$ and $v^+$ can be picked as points on the relative interiors of the two edges incident
to $v$ between $v$ and the closest intersection point.
The combinatorial complexity of a region $V_k(q)$
is at most $O(n^9)$, but the set of the boundary points $V_k(q)\cap \partial P$ consists
of only~$O(n^4)$ line segments~\cite{ADI+06}.  Given $V_{k-1}(q)\cap \partial P$,
we can compute $V_k(q)\cap \partial P$ by taking the union of the visibility regions
for~$O(n^4)$ line segments in $O(n^5)$ time~\cite{BLM02,CW12}. Instead of computing the
regions~$V_k(q)$, we iteratively maintain~$V_k(q)\cap \partial P$ for all
$k=0,\ldots , \lfloor (n-2)/4 \rfloor$ and $q\in Q$, in $O(n^2\cdot n^5)=O(n^7)$ time.

For each $k$, we find $C_k=\bigcap_{q\in Q} V_k(q)\cap {\rm int}(P)$ as follows.
First compute the intersection of the boundary segments
$B_{k-1}=\bigcap_{q\in Q} (V_{k-1}(q)\cap \partial P)$, which consists of $O(n^5)$
line segments, in $O(n^5)$ time. Then compute $C_k$ as the set of points in $P$
visible from any point in $B_{k-1}$ in $O(n^{11})$ time~\cite{ADI+06}. The binary search
tries $O(\log n)$ values of $k$, and so the total running time is $O(n^{11} \lg n)$.

Next we show that the minimum integer $k$ for which $C_k\neq \emptyset$ approximates~$R(P)$.
First, we prove that there is no $t\in {\rm int}(P)$ for which ${\rm int}(P)\subseteq V_{k-2}(t)$.
By the choice of $k$, there is no $t\in {\rm int}(P)$ for which $Q\subseteq V_{k-1}(t)$ or $\partial P \subseteq V_{k-1}(t)$ (since $Q\subset \partial P$).
Then by~\cite{Us} and Proposition~\ref{prop:independent}, ${\rm int}(P)\subseteq V_{k-2}(t)$ implies $\partial P\subseteq V_{k-1}(t)$ for any $t\in {\rm int}(P)$.

Let $s$ be an arbitrary point in $C_k$. By the choice of $k$, we have
$Q\subseteq V_k(s)$. We now show that ${\rm int}(P)\subseteq V_{k+1}(s)$.
Let $t\in {\rm int}(P)$ be an arbitrary point in the interior of $P$.
In any triangulation of $P$, point $t$ lies in some triangle~$\Delta(v_1 v_2 v_3)$, and so $t$ is directly visible from
$v_j^-$ or $v_j^+$ for $j \in \{1,2,3\}$. Since $Q\subseteq V_k(s)$,
there is a diffuse reflection path from $s$ to these boundary points
with at most~$k$ interior vertices. By appending one new segment to this
path, we obtain a diffuse reflection path from $s$ to $t$ with at most
$k+1$ interior vertices.
\end{proof}

\section{Conclusions}
\label{sec:con}

Theorem~\ref{thm:radius} establishes the upper bound of $\lfloor (n-2)/4\rfloor$ for the diffuse refection radius $R(P)$ of a simple polygon $P$ with $n$ vertices. This bound is the best possible. For a given instance $P$, we can approximate $R(P)$ up to an additive error of~2 (Theorem~\ref{thm:apx-compute-radius}). However, no polynomial-time algorithm is known for computing $R(P)$ for a given polygon $P$, or for computing the diffuse reflection center of $P$. Similarly, we know that the diffuse reflection diameter~$D(P)$ of a simple polygon with $n$ vertices is at most $\lfloor (n-2)/2\rfloor$, and this bound is the best possible~\cite{Us}, but no polynomial-time algorithm is known for computing~$D(P)$ or a diametric pair of points for a given polygon $P$.

We believe the general position assumptions about $P$ and choice of light sources
  can be avoided at the cost of more complicated analysis taking
  caution to properly handle collinear chords.

In the remainder of this section, we show that the diffuse reflection center of a polygon $P$ may not be connected or $P$-convex, and that in general there is no containment relation between the geodesic center and the diffuse reflection center. These constructions explain, in part, why it remains elusive to efficiently compute the diffuse reflection center and radius.

 \begin{figure}[htbp]
    \centering
    \includegraphics[scale=1.0]{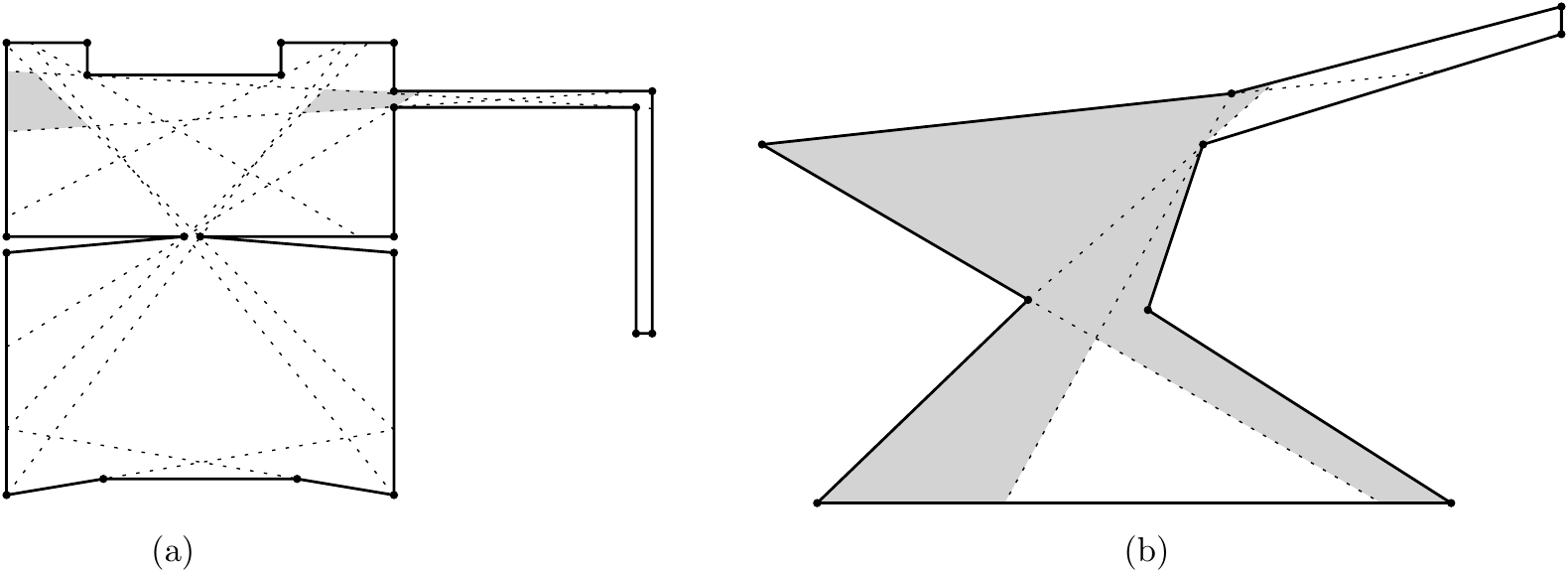}
    \caption{(a) A polygon whose diffuse reflection center is disconnected.
    (b) A polygon whose diffuse reflection center is not geodesic convex.}
    \label{fig:disconnected-center}
  \end{figure}

\smallskip
\noindent{\bf Shape of the diffuse reflection center.}
While the link center is geodesic convex and connected~\cite{LPS+88}, it turns out
that we have no such guarantees on the shape of the diffuse reflection center.
There are polygons with disconnected diffuse reflection centers (Fig.~\ref{fig:disconnected-center}(a)),
and there are polygons whose diffuse reflection centers are connected but not geodesic convex (Fig.~\ref{fig:disconnected-center}(b)).

\begin{figure}[hb!]
    \centering
    \includegraphics[scale=1.0]{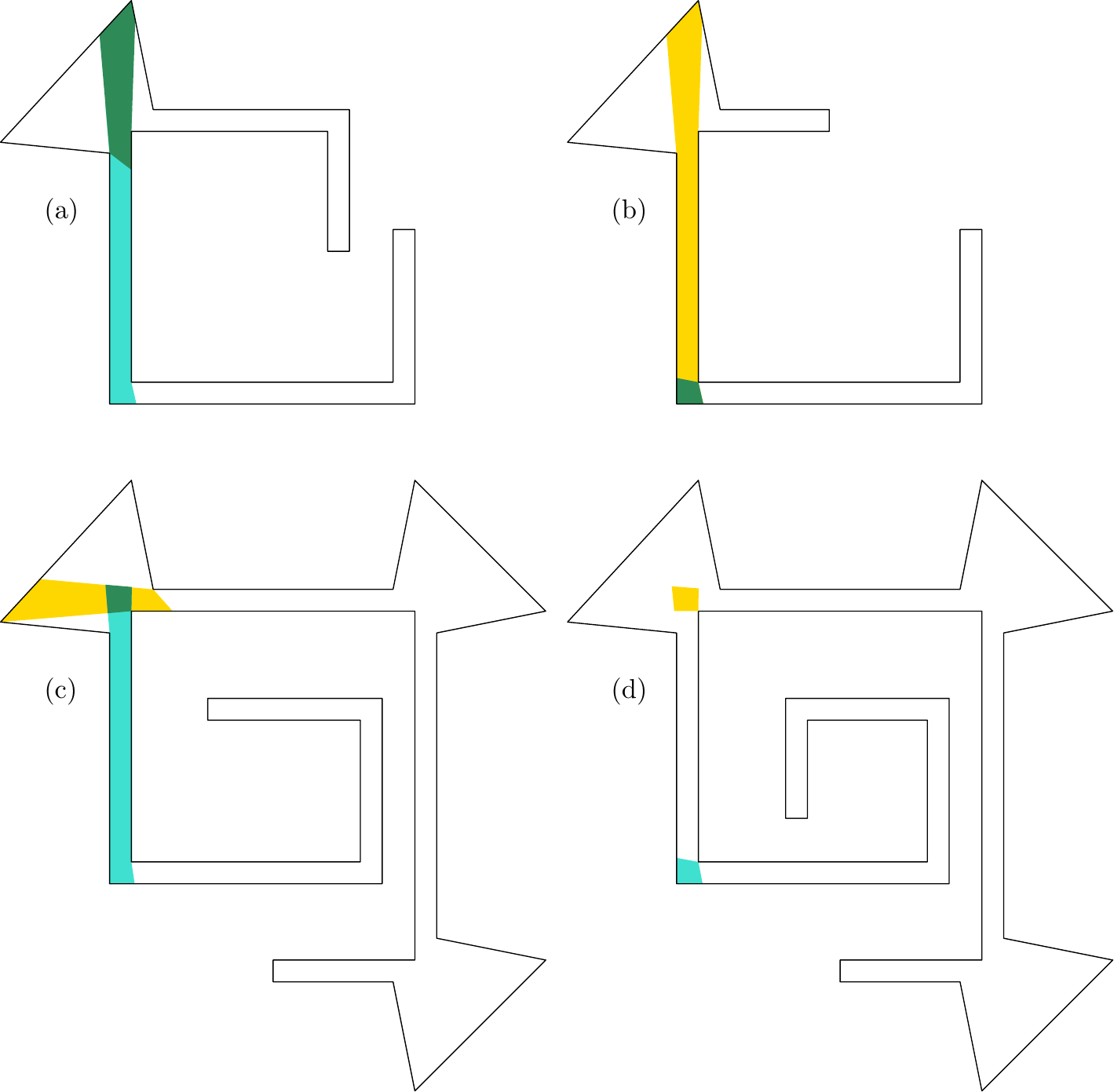}
    \caption{Examples of four inclusion relationships between the diffuse reflection and link centers.
    The diffuse center, link centers, and their intersection are colored yellow, blue, and green, respectively.
    The diffuse and link radii for the polygons in clockwise order from the upper left are 2, 2, 4, 4 and 2, 1, 3, 3, respectively.}
    \label{fig:diffuse-vs-link}
\end{figure}

Furthermore, there is no clear relationship between the two centers;
Fig.~\ref{fig:diffuse-vs-link} illustrates that there exists simple polygons
with each of the following properties:
\begin{enumerate}[(a)] \itemsep2pt
\item the diffuse reflection center is strictly contained in the link center;
\item the diffuse reflection center strictly includes the link center;
\item neither center contains the other but they are not disjoint;
\item the diffuse reflection center and the link center are disjoint.
\end{enumerate}

\bibliographystyle{spmpsci}
\bibliography{radius-arxiv}

\end{document}